\documentclass[journal]{IEEEtran}

\usepackage{amsmath}
\usepackage{amssymb,mathrsfs} 
\usepackage{graphicx}
\usepackage{amsthm}

\usepackage{subfig}

\usepackage{bm}
\usepackage{bbm}
\usepackage{amssymb}  
\usepackage{pifont}
\usepackage{array} 
\usepackage{tabularx} 
\usepackage{epstopdf}
\usepackage{dsfont}
\usepackage{color}
\usepackage{epsfig}
\usepackage{amsthm}
\makeatletter

\newcommand{\Rmnum}[1]{\expandafter\@slowromancap\romannumeral #1@}
\makeatother

\usepackage{balance}
\usepackage[utf8]{inputenc}
\usepackage[english]{babel}

\usepackage[mathscr]{euscript}
\usepackage{blkarray,booktabs, bigstrut}
\usepackage{arydshln}

\newtheorem{theorem}{Theorem}
\newtheorem{lemma}{Lemma}

\newtheorem{definition}{Definition}

\addto\captionsenglish{\renewcommand{\figurename}{Fig.}}
\usepackage[labelsep=period]{caption}

\long\def\comment#1{}

\newfont{\bbb}{msbm10 scaled 700}

\newcommand{\av}{{\bf a}}

\newcommand{\cv}{{\bf c}}

\newcommand{\ev}{{\bf e}}
\newcommand{\fv}{{\bf f}}

\newcommand{\pv}{{\bf p}}

\newcommand{\xv}{{\bf x}}
\newcommand{\yv}{{\bf y}}
\newcommand{\zv}{{\bf z}}

\newcommand{\zerov}{{\bf 0}}

\newcommand{\Am}{{\bf A}}

\newcommand{\Dm}{{\bf D}}
\newcommand{\Em}{{\bf E}}

\newcommand{\Hm}{{\bf H}}
\newcommand{\Id}{{\bf I}}

\newcommand{\Km}{{\bf K}}

\newcommand{\Wm}{{\bf W}}

\newcommand{\Ec}{{\cal E}}

\newcommand{\Gc}{{\cal G}}
\newcommand{\Hc}{{\cal H}}

\newcommand{\Nc}{{\cal N}}

\newcommand{\Rc}{{\cal R}}

\newcommand{\Vc}{{\cal V}}

\newcommand{\RNum}[1]{\uppercase\expandafter{\romannumeral #1\relax}}

\newcommand{\thetav}{\hbox{\boldmath$\theta$}}

\newcommand{\Sigmam}{\hbox{\boldmath$\Sigma$}}

\renewcommand{\arg}{{\hbox{arg}}}

\newcommand{\squeezeequ}{\medmuskip=2mu \thinmuskip=1mu \thickmuskip=3mu}

\def\LRT#1#2{\!
\raisebox{.2ex}{$
{{\scriptstyle\;#1}\atop{\displaystyle\gtrless}}
\atop
{\raisebox{-1.25ex}{$\scriptstyle\;#2$}}
$}
\!}

\begin{document}
%
\title{Stealth Attacks Against Moving Target Defense \\for Smart Grid}
%
%
%
\author{Ke Sun, 
    I\~naki Esnaola,
    and H. Vincent Poor

\thanks{
K. Sun is with the College of Computer Engineering and Science, Shanghai University, Shanghai, China
(email: ke\_sun@shu.edu.cn, cn.kesun@gmail.com).

I. Esnaola is with the School of Electrical and Electronic Engineering, University of Sheffield, Sheffield S1 3JD, UK, and also with the Department of Electrical and Computer Engineering, Princeton University, Princeton, NJ 08544 USA
(email: esnaola@sheffield.ac.uk).

H.~V. Poor is with the Department of Electrical and Computer Engineering, Princeton
University, Princeton, NJ 08544 USA 
(e-mail: poor@princeton.edu).
}}
%



\maketitle

\begin{abstract}
Data injection attacks (DIAs)  pose a significant cybersecurity threat to the Smart Grid by enabling an attacker to compromise the integrity of data acquisition and manipulate estimated states without triggering bad data detection procedures.
To mitigate this vulnerability, the moving target defense (MTD) alters branch admittances to mismatch the system information that is available to an attacker, thereby inducing an imperfect DIA construction that results in degradation of attack performance.
In this paper, we first analyze the existence of stealth attacks for the case in which the MTD strategy only changes the admittance of a single branch. 
Equipped with this initial insight, we then extend the results to the case in which multiple branches are protected by the MTD strategy.
Remarkably, we show that stealth attacks can be constructed with information only about which branches are protected, without knowledge about the particular admittance value changes.
Furthermore, we provide a sufficient protection condition for the MTD strategy via graph-theoretic tools that guarantee that the system is not vulnerable to DIAs.
Numerical simulations are implemented on IEEE test systems to validate the obtained results. 
\end{abstract}

\begin{IEEEkeywords}
Data injection attacks, stealth attacks, moving target defense, incomplete system information, graph theory
\end{IEEEkeywords}

%
\IEEEpeerreviewmaketitle

\section{Introduction}
%
%
%
%

\IEEEPARstart{D}{ata} injection attacks (DIAs) are cybersecurity threats that compromise the data integrity of state estimation procedures within the Smart Grid (SG).  
Specifically, in a DIA, an attacker comprises measurements (i.e, the data) used for state estimation by exploiting vulnerabilities in sensing or data acquisition procedures, deviating estimated states from their original values to disrupt the operation of the system \cite{giani_viking_2009,liu_false_2009}. 
The attacker also aims to ensure the stealthiness of the attacks, i.e. to bypass any bad data detection (BDD) procedure deployed by the gird operator.
For that reason,  the specification of the estimation method and the BDD procedure dictate the information required by the attacker to construct stealth attacks.
For example, for weighted least squares (WLS) estimation under a linearized system model, the attacker needs the system Jacobian matrix as prior information to be able to bypass residual-based detection \cite{liu_false_2011}. 
For minimum mean squared error estimation within a Bayesian framework, the attacker requires both the Jacobian matrix and the second-order statistics of the state variables to optimize the tradeoff between disruption and probability of detection \cite{kosut_malicious_2011, esnaola_maximum_2016, Sun_information-theoretic_2017, Sun_Stealth_2020}. 

Moving target defense (MTD) is an effective cyberthreat defense strategy that aims to change the system information or operating status to create information mismatches between operator and attacker, which in turn increases the uncertainty and complexity for attack construction \cite{MTD}. 
Specifically, for the SG paradigm, the operator can induce information mismatches by changing the admittance of specific branches \cite{rahman_2014_moving, zhang_2019_analysis, lakshminarayana_2020_cost, liu_2020_optimal, zhang_2020_hiddenness, xu_2022_robust} through flexible AC transmission system (FACTS) devices \cite{hingorani_2007_facts, zhang_2012_flexible}, or by changing the system topology \cite{morrow_2012_topology, wang_2015_effects} via transmission line switches \cite{brown_2020_transmission}. 
Beyond these techniques, mismatched information can also be achieved by introducing randomness in the measurement set \cite{rahman_2014_moving, miao_2024_robust}, e.g. changing communication routing paths \cite{heydari_2018_moving, hu_2021_network}. 
In summary, the MTD aims to make the system information that the attacker has access to inaccurate \cite{lakshminarayana_2024_survey}.
Even if the attacker identifies the inaccurate information, discarding this information renders the system information ``incomplete'' for the attacker and the attack construction is hindered. 

From the perspective of the operator, inducing inaccuracies or incompleteness in system information may aid in increasing the probability of detecting DIAs via BDD.
For example, it is shown in \cite{rahman_false_2012} that an inaccurate system Jacobian matrix indeed increases the probability of detection under residual-based detection. 
For the MTD strategy implemented by changing branch admittance, the operator must alter the admittance of a specific number of branches to eliminate the existence of stealth attacks \cite{zhang_2019_analysis, lakshminarayana_2020_cost}.  
 {\it It is important to note that these findings are based on the assumption of knowing the specific admittance changes.}
From the perspective of the attacker, inaccurate or incomplete system information imposes the requirement of learning unknown system information from the available (but incomplete in general) information; thereby increasing complexity for the attacker.
For example, the attacker can infer the system topology \cite{li_blind_2013,yu_blind_2015} or the statistical information of the state variables \cite{sun_2019_learning, sun_2023_asymptotic} from available measurements.
In that setting, the results in \cite{sun_2019_learning,sun_2023_asymptotic} establish the theoretical connection between the number of available measurements and the attack performance. 

In this paper, we address two questions that are central to the secure operation of SGs: 
\raisebox{.5pt}{\textcircled{\raisebox{-.9pt} {1}}} Can an attacker construct stealth attacks without information about the admittance changes? 
\raisebox{.5pt}{\textcircled{\raisebox{-.9pt} {2}}}  Is there an effective method for selecting the branches that are to be protected by MTD to mitigate the impact of stealth attacks?
In particular, depending on the number of MTD protected branches, we consider two cases, i.e. the Single Branch MTD case and the Multiple Branch MTD case. 
Firstly, we consider the Single Branch MTD case, in which we prove that stealth attacks can be constructed without the information about the admittance change. 
Simply knowing which branch is protected enables the attacker to circumvent the BDD and to manipulate the estimated states.
Then we extend the results in the Single Branch MTD case to the Multiple Branch MTD case, in which we provide a sufficient condition for the existence of stealth attacks when the admittance changes are unknown.
Based on this finding, we propose a sufficient condition for the MTD strategy to limit the feasibility of stealth attack constructions employing a graph-based model of the power system. 
Specifically, if the branches protected by MTD form a spanning tree, then there does not exist an attack that fulfils the sufficient condition under the Multiple Branch MTD case.
 
The contributions of the paper are summarised in the following. 
\begin{itemize}
\item {\it System Model}: It is proved in Lemma \ref{lemma:rank} that the Jacobian matrix directly built via the branch-bus incidence matrix and branch admittance matrix is rank-deficient, and that the rank is the number of state variables minus one{\footnote{Many previous works on DIAs overlook this fact. Note that this fact also shows that if a reference bus is chosen, using the angular difference w.r.t. the reference bus as state variables reduces one state variable, which guarantees that the Jacobian matrix has full rank.}}.
For this case, the estimated state variables are not unique under WLS. 
Moreover, Lemma \ref{lemma:N1} provides an alternative expression for the Jacobian matrix, which describes the Jacobian matrix as the sum of single branch related independent matrices.

\item {\it Stealth Attack Construction}: Theorem \ref{Theorem:2} and Theorem \ref{Theorem:3} show that stealth attacks still exist even without the information about the admittance change. 
Furthermore, the rationale for the existence is provided in Lemma \ref{lemma2} and Lemma \ref{lemma3}. 

\item {\it Countermeasure against Stealth Attack}: Lemma \ref{lemma4} shows a sufficient condition for MTD strategy to eliminate the existence of stealth attacks in Theorem \ref{Theorem:2} and Theorem \ref{Theorem:3}. 
The minimum number of MTD protected branches required to safeguard the system against stealth attacks is also given by Lemma \ref{lemma4}.
\end{itemize}

The rest of the paper is organized as follows. 
In Section \ref{Sec:System_Model}, the system model is introduced, in which we show the rank-deficient property of the Jacobian matrix built using branch-bus incidence matrix and branch admittance matrix. 
In Section \ref{Sec:SBMTD}, the stealth attack construction strategy is proposed for the Single Branch MTD case to prove the existence of stealth attacks. 
The Multiple Branch MTD case in considered in Section \ref{Sec:MBMTD}, in which we also show the countermeasure for the stealth attack construction strategy we considered in the previous sections. 
Section \ref{Sec:Sim} shows the numerical simulation results in Section  \ref{Sec:SBMTD} and  \ref{Sec:MBMTD}. 
The paper ends with conclusions in Section \ref{Sec:Con}.

\section{System Model} \label{Sec:System_Model}
The notation used in this paper is described below. 
Matrices are represented by boldface capital letters, such as $\Am$; and vectors are represented by boldface lowercase letters, such as $\av$.
Moreover, the vector of all zero elements is represented by $\zerov$.
Sets are represented by calligraphic font, such as $\Rc$ denoting the set of real numbers. 
The $i$-th entry of a vector $\av$ is represented by $a_i$.
The $i$-th row and the $i$-th column of a matrix $\Am$ are represented by $\Am_{i, \cdot}$ and $\Am_{\cdot,i}$, respectively; 
and the entry in the $i$-th row and $j$-th column of matrix $\Am$ is represented by $A_{i,j}$. 
The matrix that is formed by selecting a given set of columns from $\Am$ is represented by $\Am_{\cdot, \mathcal{I}}$, where $\mathcal{I}$ is the set of the column indexes. 
The branches of the power system are represented by $\mathsf{e}_{\cdot}$, where the subscript is the index of the branch. 
The $\ell_2$ norm of a vector is denoted by $\| \cdot \|_{2}$, and $| \cdot |$ represents the cardinality of a set. 

\subsection{State Estimation and Abnormality Detection}
The measurement (or observation) model for state estimation under a linearized system model is described by  
\begin{IEEEeqnarray}{c}
\yv = \Hm \xv + \zv, \label{Equ:obs}
\end{IEEEeqnarray}
where $\yv \in \Rc^{m}$ is a vector of the measurements; 
$\xv \in \Rc^{n}$ is a vector of the state variables;
$\Hm \in \Rc^{m \times n}$ is the linearized Jacobian measurement matrix, which is determined by the system topology, the admittances of the branches, and the operating point; 
$\zv \in \Rc^{m}$ is additive noise within the system. 
Note that the noise term $\zv$ is usually modeled as realizations of a vector of random variables $Z^m \in  \Rc^{m}$, with $Z^m \sim \Nc \left( \zerov, \Sigmam\right)$ and $\Sigmam$ denoting the covariance matrix \cite{grainger_power_1994, abur_power_2004}. 

Under the assumption that $\Hm$ is a full rank matrix, the WLS estimator that minimizes the estimation residual is given by 
\begin{IEEEeqnarray}{rl}
\hat{\xv} &= \arg \  \underset{\ev \in \mathcal{R}^{n}}{\min} \| \Wm \left(\yv - \Hm \ev \right)\|_{2}^2 \label{Equ:WLS}\\
&= \left( \Hm^{\sf T} \Wm \Hm \right)^{-\!1}\Hm^{\sf T}\Wm \yv \IEEEeqnarraynumspace \label{Equ:WLS0} \\
&\triangleq \Km \yv, \label{Equ:WLS1}
\end{IEEEeqnarray}
and the resulting estimation residual $r: \Rc^{m} \to \Rc^{+}$ is given by
\begin{IEEEeqnarray}{c}
r(\yv) \triangleq  \| \yv - \Hm \hat{\xv} \|_{2}^2 = \| \yv - \Hm \Km \yv \|_{2}^2, \label{Equ:Res}
\end{IEEEeqnarray}
where $\Wm \in \Rc^{m \times m}$ is a diagonal matrix whose entries are the assigned weights and are often in practical settings set to be the reciprocals of the variances of the noise \cite[(8.16)]{kay_1993_fundamentals}.
Note that when $\Hm$ is a full rank matrix, the state estimate, i.e. $\hat{\xv}$ in \eqref{Equ:WLS}, is unique for a given measurement vector $\yv$.

It is common practice in power systems for the operator of the system to compare the estimation residual with a predefined threshold to decide whether the measurements produced by the data acquisition system reflect normal operation, or on the contrary, an abnormal behavior is observed. This procedure is usually referred to as BDD.
Specifically, the operator performs the test given by
\begin{IEEEeqnarray}{c}
r \LRT{\Hc_1}{\Hc_0} \tau \label{Equ:ResTest}
\end{IEEEeqnarray}
that aims to solve the hypothesis testing problem 
\begin{IEEEeqnarray}{rl}
\Hc_0: & \textnormal{ no abnormality}\\
\Hc_1: & \textnormal{ abnormality exists}
\end{IEEEeqnarray}
to detect the existence of abnormalities, where $\tau \in \Rc$ is a predefined threshold that yields a certain probability of false alarm and probability of detection tradeoff \cite{abur_power_2004}. 

\subsection{Structure of the Jacobian Matrix}
The non-zero entries of the Jacobian matrix $\Hm$ are determined by the system topology and the ordering of the elements in the measurement vector $\yv$.
The values of the entries in $\Hm$ are determined by the admittance of the branches. 
Let $n+1$ and $l$ be the number of buses and branches in the system, respectively; ${\cal E} = \{\mathsf{e}_1, \ldots, \mathsf{e}_l\}$ be the set of all branches; $b_k$ with $k \in \{1, \ldots, l \}$ be the admittance of branch $\mathsf{e}_k$; $\Am \in \{-1, 0,1\}^{l \times (n+1)}$ be the branch-bus incidence matrix, and $\Dm  \in \Rc^{l \times l}$ be the diagonal branch admittance matrix. 
Hence for all $k \in \{1, \ldots, l\}$, the matrices $\Am$ and $\Dm$ are given by 
\begin{IEEEeqnarray}{c}
A_{ki} = \left\{ 
\begin{tabular}{cl}
 1, & \textnormal{if branch $\mathsf{e}_k$ starts from bus $i$},\\
 -1,& \textnormal{if branch $\mathsf{e}_k$ ends at bus $i$},\\
0, & \textnormal{otherwise },
\end{tabular}
\right.
\end{IEEEeqnarray}
and 
\begin{IEEEeqnarray}{c}\label{Equ:b_k}
D_{kk} = b_{k}.
\end{IEEEeqnarray}
Without loss of generality, we assume that $b_k \neq 0, \forall  k \in \{1, \ldots, l\}$, which implies that matrix $\Dm$ has full rank. 

For the linearized system model, selecting the net active power injections of the buses and the active power flow in both line directions as measurements and the bus voltage angles as state variables results in the measurement model
\begin{IEEEeqnarray}{l}
\begin{bmatrix}
\pv\\
\fv \\
-\fv
\end{bmatrix} = 
\begin{bmatrix}
\Am^{\sf T} \Dm \Am \\
 \Dm \Am \\
 - \Dm \Am 
\end{bmatrix} \thetav 
+\zv, \label{Equ:JacoPhy}
\end{IEEEeqnarray} 
where $\pv \in \Rc^{n+1}$ is the vector of net active power injections, $\fv \in \Rc^{l}$ is the vector of the active power flows, and $\thetav \in \Rc^{n+1}$ is the vector of bus voltage angles. 

However, the measurement matrix in \eqref{Equ:JacoPhy} is rank-deficient, which implies that for the WLS estimation there exist infinitely many estimates satisfying \eqref{Equ:WLS}.
The following lemma formalizes this notion. 
\begin{lemma}\label{lemma:rank}
For the linear model in \eqref{Equ:JacoPhy}, it holds that 
\begin{IEEEeqnarray}{c}
\textnormal{rank} 
\left( 
\begin{bmatrix}
\Am^{\sf T} \Dm \Am \\
 \Dm \Am \\
 - \Dm \Am 
\end{bmatrix}
\right) = n.\label{Equ:new}
\end{IEEEeqnarray}
\end{lemma}

\begin{proof}
The proof is provided in Appendix \ref{Sec:Proof0}. 
\end{proof}
To guarantee the uniqueness of the estimated states via WLS, i.e. $\Hm$ being a full rank matrix, 
we choose {\it the first bus as the reference node}, which does not decrease the generality of the model. 
As a result, it holds that 
\begin{IEEEeqnarray}{c}
\Hm = 
\begin{bmatrix}
\Am^{\sf T} \Dm \Am \\
 \Dm \Am \\
 - \Dm \Am 
\end{bmatrix}_{\cdot, \{2, \ldots, n+1 \}},  \label{Equ:JacoPhyFR}
\end{IEEEeqnarray}
which results from removing the first column.
Moreover, note that $\Hm$ in \eqref{Equ:JacoPhyFR} is equivalent to 
\begin{IEEEeqnarray}{c}
\Hm = 
\begin{bmatrix}
\Am_{\cdot, \{2, \ldots, n+1 \}}^{\sf T} \Dm\Am_{\cdot, \{2, \ldots, n+1 \}} \\
 \Dm \Am_{\cdot, \{2, \ldots, n+1 \}} \\
 - \Dm \Am_{\cdot, \{2, \ldots, n+1 \}}
\end{bmatrix},   \label{Equ:JacoPhyFR1}
\end{IEEEeqnarray}
which is obtained by removing the first column of $\Am$ \cite[pp. 259]{grainger_power_1994}. 
The following lemma provides an alternative expression for the Jacobian matrix $\Hm$. 
\begin{lemma}\label{lemma:N1}
Let $f: {\cal E} \to \{1, \ldots,n+1 \}$ be the function whose input is a branch and output is the starting bus of that branch, and let $t: {\cal E} \to \{1, \ldots,n+1 \}$ be the function whose output is the ending bus.
The Jacobian matrix $\Hm$ in \eqref{Equ:JacoPhyFR} is given by
\begin{IEEEeqnarray}{c}
\Hm = \left[\sum_{i=1}^l b_i \Hm_i \right]_{\cdot, \{2, \ldots, n+1 \}}, 
\end{IEEEeqnarray}
where $\Hm_i \in \Rc^{m \times (n+1)}$ is given by 
\begin{IEEEeqnarray}{c}
 \Hm_i = 
 \begin{blockarray}{cccccc}
   \dots & f(\mathsf{e}_i) & \dots & t(\mathsf{e}_i) & \dots& \\
 \begin{block}{[ccccc]c}
   \cdots &\cdots &\cdots & \cdots & \cdots & \vdots \\
  0 &1 &\cdots &-1 & 0 &p_{f(\mathsf{e}_i)} \\
   \vdots & \vdots & \ddots & \vdots &  \vdots &\vdots \\
  0&-1&\cdots &1 &0 &p_{t(\mathsf{e}_i)} \\
   \vdots&\vdots&\vdots &\vdots &\vdots& \vdots \\
  \cdashline{1-5}
     \vdots&\vdots&\vdots &\vdots &\vdots & \vdots \\
    0 &1 &\cdots &-1 & 0 & f_{i} \\
      \vdots&\vdots&\vdots &\vdots &\vdots &  \vdots\\
     \cdashline{1-5}
       \vdots&\vdots&\vdots &\vdots &\vdots & \vdots \\
       0 &-1 &\cdots &1 & 0 & -f_{i} \\
            \vdots&\vdots&\vdots &\vdots &\vdots & \vdots \\
 \end{block}
 \end{blockarray}
 \IEEEeqnarraynumspace \label{Equ:HmK}
\end{IEEEeqnarray}
\end{lemma}
\begin{proof}
The proof is provided in Appendix \ref{Sec:Proof1}. 
\end{proof}

Although only the case in which the branches with real admittances is considered, i.e. $\Dm  \in \Rc^{l \times l}$, the results in this subsection can be readily extended to the case of complex admittances, see \cite{kettner_2017_properties}.

\subsection{Data Injection Attacks}

DIA is a kind of data integrity attack that corrupts the measurements acquired by the SCADA system, aiming to mislead the operator. 
Specifically, the measurement model under a DIA is described by 
\begin{IEEEeqnarray}{c}
\label{eq:attack_obs_mod}
\yv_a = \Hm \xv + \zv +\av, 
\end{IEEEeqnarray}
where $\yv_a \in \Rc^{m} $ is a vector of measurements that are under attack and $\av \in \Rc^{m}$ is a vector that represents the attack. 
Note that the objective of the attacker is two-fold: corrupting the measurements to disrupt the state estimation and to bypass the BDD stealthily. 
To that end, an attack that bypasses the BDD and corrupts the estimated states is defined as a {\it stealth attack}.

\begin{definition}\label{Def:S}
An attack vector $\av$ is a stealth attack under the Jacobian matrix $\Hm$ when the following conditions are simultaneously satisfied:
\begin{itemize}
\item $\Km \yv \neq \Km \yv_a$
\item $r(\yv) = r(\yv_a)$,
\end{itemize}
where $\Km$ and $r$ are defined in \eqref{Equ:WLS1} and \eqref{Equ:Res}, respectively. 
\end{definition}

It is proved in \cite{liu_false_2011} that attacks that satisfy the constraint $\av =\Hm \cv$ with $\cv \in \Rc^{n}$ and $\cv \neq \zerov$ do not change the residual test in \eqref{Equ:ResTest}, and therefore, can arbitrarily deviate the estimated states. 
The following theorem states the result.

\begin{theorem}{\textnormal{\cite[Theorem 3.1]{liu_false_2011}}}\label{Theorem:Liu}
Let $\Hm$ be the Jacobian matrix for the observation model in \eqref{eq:attack_obs_mod}. 
Then the attack vector given by 
\begin{IEEEeqnarray}{c}
\av = \Hm \cv \label{Equ:LiuAttack}
\end{IEEEeqnarray}
does not change the residual and is a stealth attack for any vector $\cv \in \{ \Rc^{n}\setminus \zerov\}$.
\end{theorem}

\subsection{Moving Target Defense}
MTD provides a cyberattack defense strategy whereby the operator of the system willingly changes the system information to create a mismatch between the real system information and the knowledge that the attacker possesses about it. 
In doing so, the operator increases the uncertainty and increases the complexity of the stealth attack construction for the attacker \cite{zhuang_2014_towards}.
Specifically, for DIAs, the operator of the system can alter the impedances (or equivalently admittances) of branches via FACTS devices \cite{rogers_2008_some}, which induces information mismatch in the Jacobian matrix. 
{\it Here we assume that the attacker is unable to detect or observe changes in the Jacobian matrix, implying that the attacker has access only to knowledge about the Jacobian matrix before the implementation of the mismatch induced by the MTD strategy.}

Let ${\cal E_{MTD}} \subset {\cal E}$ be the set of branches that are equipped with FACTS devices, and $\Dm' \in \Rc^{l \times l} $ be the diagonal branch admittance matrix of the system protected by MTD.
Then, for all $i \in \{1, \ldots, l \}$, the matrix $\Dm'$ satisfies 
\begin{IEEEeqnarray}{c}\label{Equ:DBAMTD}
D_{ii}' = \left\{ 
\begin{tabular}{cc}
$b_i$, & \textnormal{if $\mathsf{e}_i \not\in {\cal E_{MTD}}$} ,\\
$b_i + \Delta b_i $, & \textnormal{if $\mathsf{e}_i \in {\cal E_{MTD}}$}, \\
\end{tabular}
\right. 
\end{IEEEeqnarray}
where $\Delta b_i  \in \Rc $ is the change in admittance of branch $\mathsf{e}_i$ 
and the Jacobian matrix of the  system protected by MTD is given by 
\begin{IEEEeqnarray}{c}\label{Equ:Jac_MTD}
\Hm' = 
\begin{bmatrix}
\Am^{\sf T} \Dm' \Am \\
 \Dm' \Am \\
 - \Dm' \Am 
\end{bmatrix}_{\cdot, \{2, \ldots, n+1 \}}.
\end{IEEEeqnarray}
We extend the stealth attack construction in Definition \ref{Def:S}, by also adding the mismatch information requirement that guarantees that the attack can bypass the BDD in a stealthy manner and corrupt the estimated states. 
We refer to these attacks as {\it MTD-resilient stealth attacks}.
\begin{definition}\label{Def:SMTD}
An attack vector $\av$ is an MTD-resilient stealth attack under a given MTD strategy (or equivalently the corresponding Jacobian matrix $\Hm'$) when the following conditions are simultaneously satisfied:
\begin{itemize}
\item $\Km' \yv \neq \Km' \yv_a$
\item $r'(\yv) = r'(\yv_a)$,
\end{itemize}
where $\Km'$ and $r'$ are defined as 
\begin{IEEEeqnarray}{rll}
\Km' &\triangleq  \left( \Hm'^{\sf T} \Wm \Hm' \right)^{-\!1}\Hm'^{\sf T}\Wm, &\quad \textnormal{and} \label{Equ:WLSN}
\\
r'(\yv) &\triangleq  \| \yv - \Hm' \hat{\xv} \|_{2}^2 = \| \yv - \Hm' \Km' \yv \|_{2}^2. &\label{Equ:ResN}
\end{IEEEeqnarray}
\end{definition}

\section{Single Branch Moving Target Defense}\label{Sec:SBMTD}
In this section, we first consider a special case of the MTD for which the operator can only change the admittance of a single branch, i.e. $|{\cal E_{MTD}}|=1$. 
This case is referred to as {\it Single Branch MTD}. 

It is intuitive to assume that the efficacy of the MTD is determined by the amount of branches whose admittance can be changed and by the magnitude of the admittance change, i.e. by $|{\cal E_{MTD}}|$ and by $\Delta b_i$ with $i \in \{1, \ldots,l \}$.
When the operator is limited to changing the admittance of a single branch, it is not straightforward to assume that {\it MTD-resilient stealth attacks} can be constructed without information about $\Delta b_i$. 
However, the following theorem proves that such attacks can indeed be constructed, and moreover, the implementation does not require knowledge about $\Delta b_i$. 

Note that there is no branch starting and ending at the same bus in the power system, hence there are only two cases that need to be considered: either the protected branch connects the reference bus (i.e. bus with index one) or it does not connect the reference bus. 

\begin{theorem}\label{Theorem:2}
Assume that ${\cal E_{MTD}} = \{ \mathsf{e}_k\}$ with $k \in \{1, \ldots,l \}$.
If $f(\mathsf{e}_k) \neq 1$ and $t(\mathsf{e}_k) \neq 1$, then for all vectors $\cv \in \{ \Rc^{n}{\setminus}\zerov\}$ such that 
\begin{IEEEeqnarray}{c}
c_{f(\mathsf{e}_k)-1} = c_{t(\mathsf{e}_k)-1} , \label{Equ:SBMTDCon1} \IEEEeqnarraynumspace
\end{IEEEeqnarray}
the attack vectors satisfying 
\begin{IEEEeqnarray}{c}\label{Equ:SBMTD}
\av = \Hm \cv 
\end{IEEEeqnarray}
are MTD-resilient stealth attacks under the given Single Branch MTD strategy, where function $f$ and $t$ are defined in Lemma \ref{lemma:N1}.  

If $f(\mathsf{e}_k) = 1$ or $t(\mathsf{e}_k) = 1$, then for all vectors $\cv \in \{ \Rc^{n}{\setminus}\zerov\}$ such that 
\begin{subequations}\label{Equ:SBMTDCon2}
\begin{IEEEeqnarray}{c}
c_{f(\mathsf{e}_k)-1} = 0, \quad \textnormal{ if} \ t(\mathsf{e}_k) \ = 1, \\
c_{t(\mathsf{e}_k)-1}  = 0, \quad \textnormal{ if} \ f(\mathsf{e}_k) = 1, 
\end{IEEEeqnarray}
\end{subequations}
the attacks satisfying \eqref{Equ:SBMTD} are MTD-resilient stealth attacks under the given Single Branch MTD strategy. 
\end{theorem}
\begin{proof}
The proof is provided in Appendix \ref{Sec:Proof3}.
\end{proof}

Theorem \ref{Theorem:2} establishes a sufficient condition for attacks to be MTD-resilient stealth attacks under the Single Branch MTD. 
{\it Importantly, the attack in Theorem \ref{Theorem:2} relies on the Jacobian matrix before MTD, i.e. on $\Hm$, rather than the Jacobian matrix after the implementation of the MTD.}
In view of this, if the attacker knows which branch is protected by the Single Branch MTD strategy, the attacker can construct MTD-resilient stealth attacks. 
{\it It is worth mentioning that the result in Theorem \ref{Theorem:2} holds irrespective of the value of $\Delta b_k$}, which implies the admittance value that is changed by the FACTS device does not impact the way of constructing MTD-resilient stealth attacks. 
This suggests that the operator needs to secure the information about location of the MTD protected branch. 

The following lemma shows that the stealthiness of the attacks in Theorem \ref{Theorem:2} relies on the fact that the attack does not change the power flow on the single branch protected by MTD. 
\begin{lemma}\label{lemma2}
Assume that ${\cal E_{MTD}} = \{ \mathsf{e}_k\}$ with $k \in \{1, \ldots,l \}$ and that the measurements are sorted as in \eqref{Equ:JacoPhy}. 
The attack $\av$ constructed via Theorem \ref{Theorem:2} does not alter the power flow on $\mathsf{e}_k$, i.e. 
\begin{IEEEeqnarray}{c}
a_{n+1+k} = 0 \quad \textnormal{and} \quad a_{n+1+l+k} = 0,
\end{IEEEeqnarray}
where $a_{n+1+k}$ and  $a_{n+1+l+k}$ are the attack injection to the power flow on $ \mathsf{e}_k$.
\end{lemma}

\begin{proof}
The proof is provided in Appendix \ref{Sec:Proof4}.
\end{proof}

The intuition behind Lemma \ref{lemma2} is outlined in the following. 
If the attacker is aware of which branch is protected by the Single Branch MTD strategy, 
the attacker can refrain from injecting any error into the measurement from that branch. 
Adding the fact that the attacker knows the admittance of the rest of the branches, the attacker can inject error into the measurements on these branches without triggering the BDD. 

\section{Multiple Branch Moving Target Defense and Countermeasure}\label{Sec:MBMTD}
\subsection{Multiple Branch Moving Target Defense}
In this section, we extend the results in Section \ref{Sec:SBMTD} to the Multiple Branch MTD case, that is, the case for which it holds that $|{\cal E_{MTD}}|>1$.
For the Multiple Branch MTD case, the operator can change the admittance of more than one branch, which suggests that it will be, in principle, more effective against DIAs. 
However, the following theorem shows that MTD-resilient stealth attacks are possible and describes the conditions for existence. 
As in the Single Branch MTD case, there are only two cases that need consideration: no protected branch connects the reference bus, and there exists a protected branch connecting the reference bus. 
The theorem below describes the existence for both cases.

\begin{theorem}\label{Theorem:3}
Assume that 
\begin{IEEEeqnarray}{c}\label{Equ:MBMTDCon0}
 \forall \mathsf{e}  \in {\cal E_{MTD}},  \quad f(\mathsf{e}) \neq 1 \ \textnormal{and} \ t(\mathsf{e}) \neq 1.
\end{IEEEeqnarray}
Then for all vectors $\cv \in \{ \Rc^{n}{\setminus}\zerov\}$ such that 
\begin{IEEEeqnarray}{c}
 \forall \mathsf{e}  \in {\cal E_{MTD}},  \quad c_{f(\mathsf{e})-1} = c_{t(\mathsf{e})-1},  \label{Equ:MBMTDCon1} 
\end{IEEEeqnarray}
the attacks satisfying 
\begin{IEEEeqnarray}{c}\label{Equ:MBMTD}
\av = \Hm \cv 
\end{IEEEeqnarray}
are MTD-resilient stealth attacks under the given Multiple Branch MTD strategy, where function $f$ and $t$ are defined in Lemma \ref{lemma:N1}. 

Assume that 
\begin{IEEEeqnarray}{c}\label{Equ:MBMTDCon3N}
\exists \mathsf{e}  \in {\cal E_{MTD}}, \quad f(\mathsf{e}) = 1 \ \textnormal{or} \ t(\mathsf{e}) = 1. 
\end{IEEEeqnarray}
Then for all vectors $\cv \in \{ \Rc^{n}{\setminus}\zerov\}$ such that  $ \forall \mathsf{e}  \in {\cal E_{MTD}}$,
\begin{subequations}\label{Equ:MBMTDCon2}
\begin{IEEEeqnarray}{rl}
c_{f(\mathsf{e})-1} = 0, \quad &\textnormal{ if} \ t(\mathsf{e}) \ = 1, \label{Equ:MBMTDCon21} \\
c_{t(\mathsf{e})-1} \ = 0, \quad &\textnormal{ if} \ f(\mathsf{e}) = 1,\label{Equ:MBMTDCon22} \\
c_{f(\mathsf{e})-1} = c_{t(\mathsf{e})-1},   \quad &\textnormal{ if} \ f(\mathsf{e}) \ \neq 1 \textnormal{ and }\ t(\mathsf{e}) \ \neq 1, \label{Equ:MBMTDCon23} \IEEEeqnarraynumspace
\end{IEEEeqnarray}
\end{subequations}
the attacks satisfying \eqref{Equ:MBMTD} are MTD resislent stealth attacks under the given Multiple Branch MTD strategy. 

\end{theorem}

\begin{proof}
The proof is provided in Appendix \ref{Sec:Proof5}.
\end{proof}

Theorem \ref{Theorem:3} provides a sufficient condition for the attacks to be MTD-resilient stealth attacks under the Multiple Branch MTD. 
As in Theorem \ref{Theorem:2},  {\it the result in Theorem \ref{Theorem:3} holds for the Jacobian matrix before MTD and for any value of $\Delta b_i, i\in \{1, \ldots,m \}$.}
Specifically, constructing stealth attacks only necessitates knowledge regarding the protected branches and the Jacobian matrix before MTD.
This underscores the importance for the operator to secure information pertaining to the MTD protected branches.

The following lemma shows that the stealthiness of the attacks in Theorem \ref{Theorem:3} also relies on the fact that the attacks do not change the power flow over the branches protected by MTD, which is the same condition as in Lemma \ref{lemma2} for the Single Branch MTD case.

\begin{lemma}\label{lemma3}
Let $b: {\cal E} \to \{1, \ldots,l \}$ be the function whose input is a branch and output is the index of that branch.
Assume that the measurements are sorted as \eqref{Equ:JacoPhy}. 
The attack $\av$ constructed via Theorem \ref{Theorem:3} does not corrupt the power flow on the branch belongs to $ {\cal E_{MTD}}$, i.e. 
\begin{IEEEeqnarray}{c}
 \forall  \mathsf{e} \in {\cal E_{MTD}}, \quad a_{n+1+b(\mathsf{e})} = 0 \ \ \textnormal{and} \ \ a_{n+1+l+b(\mathsf{e})} = 0, \IEEEeqnarraynumspace \label{Equ:MBMTD_S}
\end{IEEEeqnarray}
where $a_{n+1+b(\mathsf{e})}$ and  $a_{n+1+l+b(\mathsf{e})}$ are the injected attack to the power flow on $ \mathsf{e}$.
\end{lemma}

\begin{proof}
The proof is provided in Appendix \ref{Sec:Proof6}.
\end{proof}

\subsection{Countermeasuring MTD-Resilient Stealth Attacks}
Another significant insight from Lemma \ref{lemma3} is that if the power flows on all branches are not changed, the system can be fully protected, i.e. MTD-resilient stealth attack in Theorem \ref{Theorem:3} cannot be constructed. 
This insight is formally articulated in the following lemma, in which a countermeasure for the MTD-resilient stealth attacks in Theorem \ref{Theorem:3} is proposed via graph-theoretic tools, specifically using a Spanning Tree \cite[Chapter 2.2]{west_2001_introduction} based argument. 

\begin{lemma}\label{lemma4}
Let $\Gc(\Vc,\Ec)$ be the undirected graph formed by the set of buses $\Vc$ and the set of branches $\Ec$ with branch-bus incidence matrix $\Am$. 
If the subgraph $\Gc'(\Vc,{\cal E_{MTD}})$ is a spanning tree of $\Gc(\Vc,\Ec)$, there does not exist a vector $\cv \in \{ \Rc^{n}{\setminus}\zerov\}$ satisfying the condition in \eqref{Equ:MBMTDCon1} or \eqref{Equ:MBMTDCon2}. 
\end{lemma}
\begin{proof}
The proof is provided in Appendix \ref{Sec:Proof7}.
\end{proof}

Given the fact that any power system can be modeled by a connected and undirected graph, the spanning tree of the power system can always be constructed \cite[Chapter 2.2]{west_2001_introduction}. 
The spanning tree with the minimum number of edges (or branches) is the {\it Minimum Spanning Tree (MST)}. 
Note that for the MST problem, different weights can be assigned to different edges. 
Here, due to the fact that the results in Theorem \ref{Theorem:3} hold for any value of $\Delta b_i,  \i \in \{1, \ldots, l \}$, the weight on a particular edge (or branch) can be understood as the difficulty of changing the admittance of that branch. 
There are various algorithms for solving the MST problem, which are mainly greedy algorithms, such as the Prim's algorithm and Kruskal's algorithm \cite{graham_1985_history}.

It is worth mentioning that the results in Theorem \ref{Theorem:3} and Lemma \ref{lemma4} hold for any value of $\Delta b_i, i \in \{1, \dots, l\}$. 
Hence, the operator can just alter the admittance of branches belonging to ${\cal E_{MTD}}$ with small perturbations to guarantee the effectiveness of the MTD against stealth attacks, which suggests that there are defense strategies available to the operator against these types of attacks.

\section{Simulations}\label{Sec:Sim}
In this section, we evaluate the performance of the MTD-resilient stealth attacks constructed using Theorem \ref{Theorem:2} and Theorem \ref{Theorem:3} and of the countermeasure strategy presented in Lemma \ref{lemma4} within practical state estimation scenarios. 
Specifically, we use the IEEE test systems, whose parameters and topology are obtained from MATPOWER \cite{zimmerman_matpower:_2011}.
Without loss of generality, we arrange the measurements in the vector $\yv$ in \eqref{Equ:JacoPhy}, use the branch indexes provided in the test case data file of MATPOWER, and select the first bus (i.e. the bus with index one) as the reference bus.

\subsection{Assessment of the Single Branch MTD Case}
For the Single Branch MTD case, we focus on the scenario that assumes that the attacker lacks knowledge about which branch is protected. In this setting, the attacker is thus compelled to construct the attack as $\av = \Hm\cv$, i.e. without any MTD robustness in the construction. 
Specifically, we aim to evaluate the effectiveness of Single Branch MTD in such circumstances. 
To that end, we randomly generate the injected error $\cv$ following a standard multivariate uniform distribution, and compare the probability of detection with and without Single Branch MTD.
To model an additive white Gaussian noise scenario in the observation model, we set the covariance matrix of the noise $\Sigmam$ to $\sigma^2 \Id$ with $\sigma=3$ and $\Id$ denoting the indentity matrix.
Additionally, the threshold $\tau$ is set to $\sigma^2 \chi^2_{m-n, 0.95}$, which sets the probability of false alarm to $0.05$.

Fig. \ref{Fig:SBMTD_30} shows the increase in probability of detection when the operator protects different branches in the IEEE 30-Bus test system, for which we randomly generate $20 \times 41$ attacks (the number of branches is $41$ in IEEE 30-Bus system) and set the changes in the admittance of branch $\mathsf{e}_i, \forall i \in \{1, \ldots, 41 \}$ to 
\begin{IEEEeqnarray}{c}
\Delta b_i = \alpha b_i.
\end{IEEEeqnarray}
Notice that the the increase in probability of detection is minor, which indicates that the Single Branch MTD has limited capability to protect against the DIAs. 
It is worth mentioning that the increase in probability of detection increases when the magnitude of the corruption injected into the buses connected by the protected bus increases. 
However, the operator is likely to detect the abnormality in system if the estimated states deviate significantly from the normal values.

\begin{figure}[t!]
\centering
\includegraphics[scale=0.465]{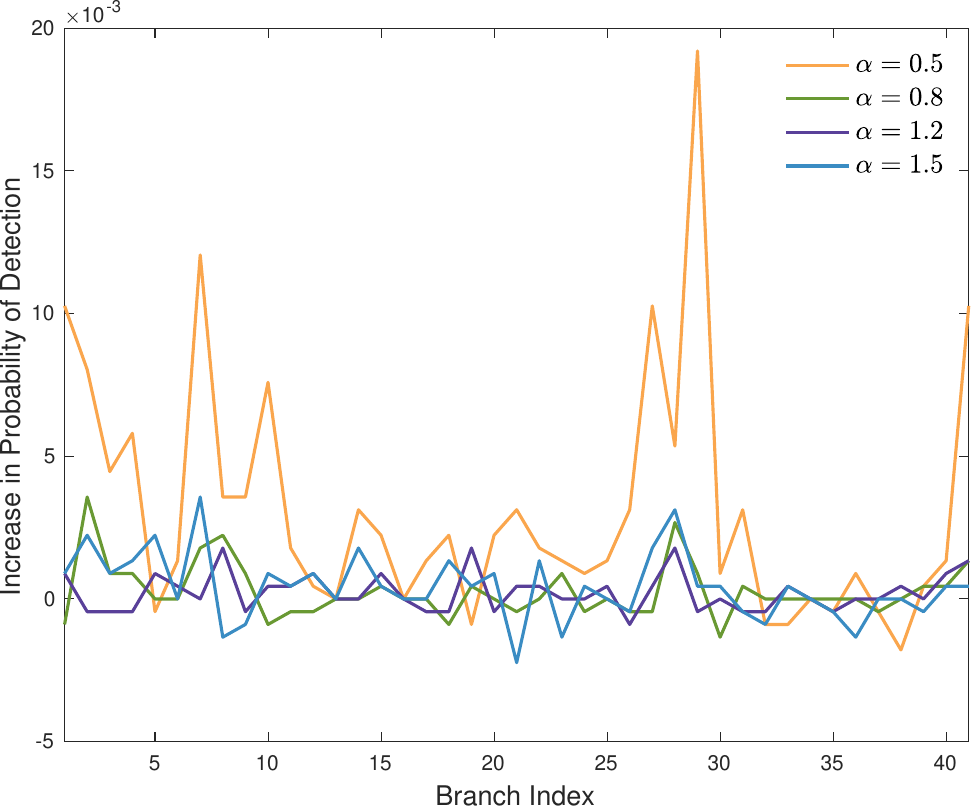}
\caption{Increase in probability of detection when the operator protects different branches in the IEEE 30-Bus test system.}
\label{Fig:SBMTD_30}
\end{figure}

\subsection{Assessment of Multiple Branch MTD Case}
For the Multiple Branch MTD case, intuition suggests that the probability of constructing MTD-resilient stealth attacks via Theorem \ref{Theorem:3} decreases when the number of protected branches increases. 
To validate this intuition, we use the ratio between the number of protected branches and the total number of branches in the system as the $x$-axis, and randomly generate the index of the protected branches for $1,000$ realizations to estimate the probability of MTD-resilient stalth attacks existence.
The result is presented in Fig. \ref{Fig:MBMTD_success}.
It can be found that if the operator protects less than $40\%$ of the branches, the MTD-resilient stealth attack in Theorem \ref{Theorem:3} always exists. 
And as the number of the protected branches increases, the probability of existence starts to fall and reaches zero when all the branches are protected. 

Another interesting observation that can be distilled from Fig. \ref{Fig:MBMTD_success} is that as the size of the system increases, the ratio of branches that need to be protected also increases, i.e. large scale systems are more vulnerable to MTD-resilient stealth attacks.  
From a practical standpoint,  in large scale systems, the operator must protect a higher proportion of branches to achieve an equivalent level of protection.

\begin{figure}[t!]
\centering
\includegraphics[scale=0.48]{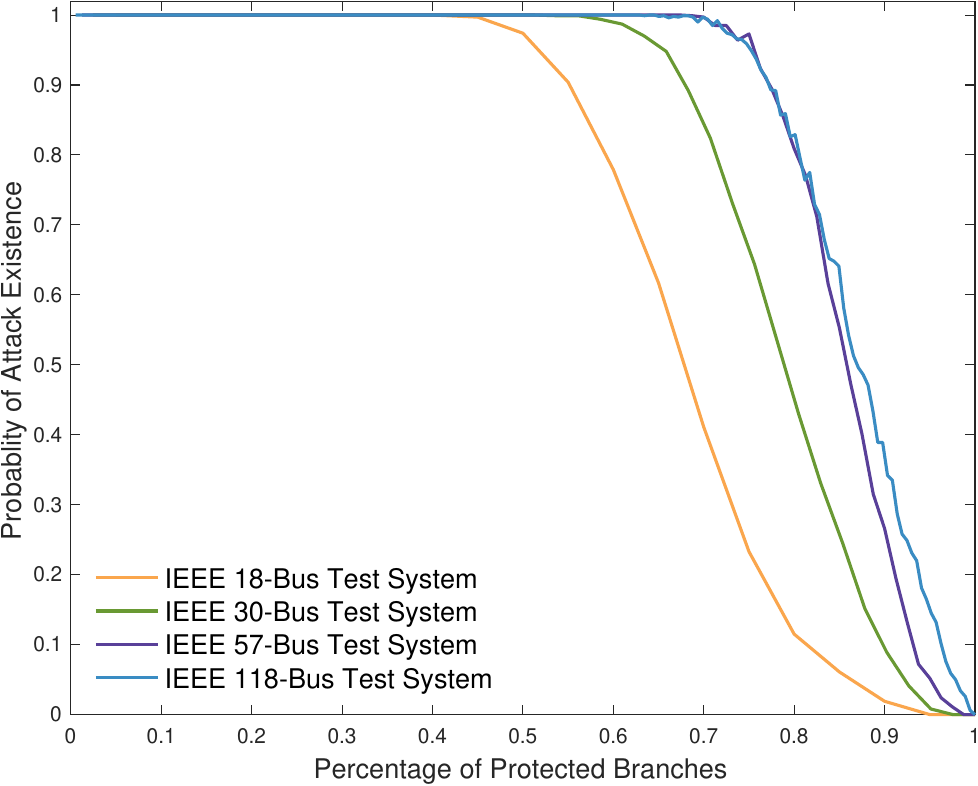}
\caption{Probability of MTD-resilient stealth attack in Theorem \ref{Theorem:3} existence when the operator randomly chooses the protected branches.}
\label{Fig:MBMTD_success}
\end{figure}

\subsection{Simulation of Countermeasure}
Fig. \ref{Fig:14_MST} depicts in red the branches in the MST of IEEE 14-Bus test system \cite{leon_2020_quadratically} when the branches are equally weighted. 
Note that Prim's algorithm is used to find the MST of the grid. 
As depicted, the MST starts from the reference bus, i.e. Bus 1, and spreads to the buses with larger indices. 
Table \ref{Tab:MST} shows the percentage of branches in the MST when the branches are equally weighted, with the percentage defined as the ratio between the number of branches in the MST and the total number of branches in the grid. 
It can be observed that the operator has to protect around $62\% - 72\%$ of the branches to form an MST, which guarantees that the stealth attack constructions described in Theorem \ref{Theorem:2} and Theorem \ref{Theorem:3} are not feasible. 

\begin{figure}[t!]
\centering
\includegraphics[scale=0.35]{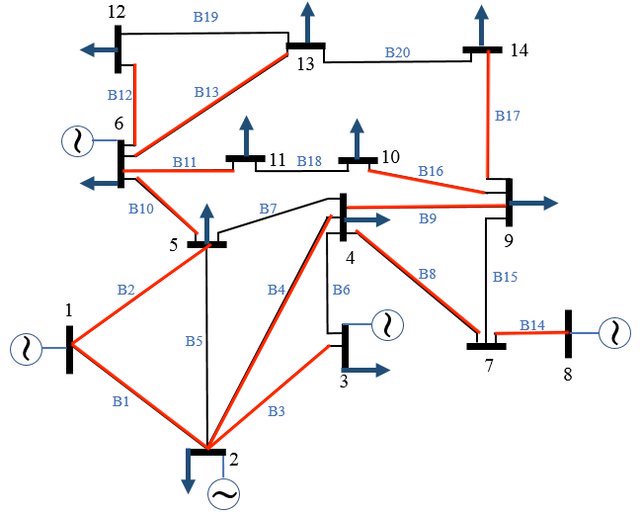}
\caption{Branches in the MST of IEEE 14-Bus test system when the branches are equally weighted.}
\label{Fig:14_MST}
\end{figure}

\begin{table}
\renewcommand{\arraystretch}{1.3}
\caption{Number of branches in the MST with equal weights}
\label{Tab:MST}
\centering
\begin{tabular}{c|c|c}
\hline
\bfseries Test System & \bfseries Number of Branches in MST & \bfseries Percentage \\
\hline
IEEE 14-Bus System & 13 & 0.65\\
\hline
IEEE 30-Bus System & 29 & 0.707\\
\hline
IEEE 57-Bus System & 56 & 0.7\\
\hline
IEEE 118-Bus System & 117 & 0.629 \\
\hline
IEEE 300-Bus System & 299 & 0.7275 \\
\hline
\end{tabular}
\end{table}

Table \ref{Tab:MST1} shows the number of branches in the MST when the weight of the branch is modeled by the function $w: {\cal E} \to \Rc$ such that 
\begin{IEEEeqnarray}{c}\label{Equ:weight}
w\left( \mathsf{e}_k \right) \propto b_k
\end{IEEEeqnarray}
with $b_k$ in \eqref{Equ:b_k}. This choice models the fact that it is more difficult to alter branches with large admittance values.
It is interesting to note that there is no difference between the case with equal weights and the case with the weights satisfying \eqref{Equ:weight}. 
This observation holds even for the case in which the weight of branch is randomly assigned. 
This suggests that the topology of the grid, not the electrical characteristics describing the system, is what governs the vulnerability of the system to stealth attacks.
This can be perhaps explained by the fact that the tree-structure of the transmission or distribution systems and its relevance to the minimum cost principle of electricity transmission \cite{li_2011_probabilistic}. 

\begin{table}
\renewcommand{\arraystretch}{1.3}
\caption{Number of branches in the MST with branch weight satisfying \eqref{Equ:weight}}
\label{Tab:MST1}
\centering
\begin{tabular}{c|c|c}
\hline
\bfseries Test System & \bfseries Number of Branches in MST & \bfseries Percentage \\
\hline
IEEE 14-Bus System & 13 & 0.65\\
\hline
IEEE 30-Bus System & 29 & 0.707\\
\hline
IEEE 57-Bus System & 56 & 0.7\\
\hline
IEEE 118-Bus System & 117 & 0.629 \\
\hline
IEEE 300-Bus System & 299 & 0.7275 \\
\hline
\end{tabular}
\end{table}

\section{Conclusion}\label{Sec:Con}
In this paper, we have demonstrate the feasibility of stealth attacks for power systems protected by the MTD strategy. 
To do this, we proved that an attacker only needs the information about which branches are protected in order to successfully construct an MTD-resilient stealth attack.
We also have shown that the stealth of these attacks relies on the fact that the attacks are undetectable if the attacker does not compromise the power flow of the protected branches. 
Based on this observation, we have proposed a countermeasure by ensuring that the branches protected by MTD form a spanning tree in the system. 
Simulations on the IEEE test systems reveal that protecting a single branch provides limited protection in way of increasing the probability of attack detection. 
Moreover, it is concluded from the numerical analysis that the spanning tree of the power system is unlikely to change even when different weights are assigned to the branches, which suggests that the vulnerability of the system to stealth attacks is governed by the topology of the system and not by its electrical characteristics. 

\bibliographystyle{IEEEtran}
\bibliography{reference}

\begin{thebibliography}{10}
\providecommand{\url}[1]{#1}
\csname url@samestyle\endcsname
\providecommand{\newblock}{\relax}
\providecommand{\bibinfo}[2]{#2}
\providecommand{\BIBentrySTDinterwordspacing}{\spaceskip=0pt\relax}
\providecommand{\BIBentryALTinterwordstretchfactor}{4}
\providecommand{\BIBentryALTinterwordspacing}{\spaceskip=\fontdimen2\font plus
\BIBentryALTinterwordstretchfactor\fontdimen3\font minus
  \fontdimen4\font\relax}
\providecommand{\BIBforeignlanguage}[2]{{%
\expandafter\ifx\csname l@#1\endcsname\relax
\typeout{** WARNING: IEEEtran.bst: No hyphenation pattern has been}%
\typeout{** loaded for the language `#1'. Using the pattern for}%
\typeout{** the default language instead.}%
\else
\language=\csname l@#1\endcsname
\fi
#2}}
\providecommand{\BIBdecl}{\relax}
\BIBdecl

\bibitem{giani_viking_2009}
A.~Giani, S.~Sastry, K.~H. Johansson, and H.~Sandberg, ``The {VIKING} project:
  {An} initiative on resilient control of power networks,'' in \emph{Proc. 2nd
  {Int.} {Symp.} on {Resilient} {Control} {Syst.}}, Idaho Falls, ID, USA, Aug.
  2009, pp. 31--35.

\bibitem{liu_false_2009}
Y.~Liu, P.~Ning, and M.~K. Reiter, ``False {data} {injection} {attacks}
  {against} {state} {estimation} in {electric} {power} {grids},'' in
  \emph{Proc. ACM Conf. on Computer and Communications Security}, Chicago, IL,
  USA, Nov. 2009, pp. 21--32.

\bibitem{liu_false_2011}
------, ``False {data} {injection} {attacks} {against} {state} {estimation} in
  {electric} {power} {grids},'' \emph{ACM Trans. Inf. Syst. Secur.}, vol.~14,
  no.~1, pp. 13:1--13:33, Jun. 2011.

\bibitem{kosut_malicious_2011}
O.~Kosut, L.~Jia, R.~J. Thomas, and L.~Tong, ``Malicious {data} {attacks} on
  the {smart} {grid},'' \emph{IEEE Trans. Smart Grid}, vol.~2, no.~4, pp.
  645--658, Dec. 2011.

\bibitem{esnaola_maximum_2016}
I.~Esnaola, S.~M. Perlaza, H.~V. Poor, and O.~Kosut, ``Maximum {distortion}
  {attacks} in {electricity} {grids},'' \emph{IEEE Trans. Smart Grid}, vol.~7,
  no.~4, pp. 2007--2015, Jul. 2016.

\bibitem{Sun_information-theoretic_2017}
K.~Sun, I.~Esnaola, S.~M. Perlaza, and H.~V. Poor, ``{Information}-{theoretic}
  {attacks} in the {smart} {grid},'' in \emph{Proc. {IEEE} {Int.} {Conf.} on
  {Smart} {Grid} {Commun.}}, Dresden, Germany, Oct. 2017, pp. 455--460.

\bibitem{Sun_Stealth_2020}
------, ``Stealth attacks on the smart grid,'' \emph{IEEE Trans. Smart Grid},
  vol.~11, no.~2, pp. 1276--1285, Mar. 2020.

\bibitem{MTD}
{U.S. Department of Homeland Security}, ``Moving target defense,''
  https://www.dhs.gov/science-and-technology/csd-mtd, Jan. 2023.

\bibitem{rahman_2014_moving}
M.~A. Rahman, E.~Al-Shaer, and R.~B. Bobba, ``Moving target defense for
  hardening the security of the power system state estimation,'' in \emph{Proc.
  First ACM Workshop on Moving Target Defense}, Arizona, Scottsdale, USA, Nov
  2014, pp. 59--68.

\bibitem{zhang_2019_analysis}
Z.~Zhang, R.~Deng, D.~K.~Y. Yau, P.~Cheng, and J.~Chen, ``Analysis of moving
  target defense against false data injection attacks on power grid,''
  \emph{IEEE Trans. Inf. Forensics Secur.}, vol.~15, pp. 2320--2335, 2019.

\bibitem{lakshminarayana_2020_cost}
S.~Lakshminarayana and D.~K.~Y. Yau, ``Cost-benefit analysis of moving-target
  defense in power grids,'' \emph{IEEE Trans. Power Systems}, vol.~36, no.~2,
  pp. 1152--1163, Jul. 2020.

\bibitem{liu_2020_optimal}
B.~Liu and H.~Wu, ``Optimal {D-FACTS} placement in moving target defense
  against false data injection attacks,'' \emph{IEEE Trans. Smart Grid},
  vol.~11, no.~5, pp. 4345--4357, Sep. 2020.

\bibitem{zhang_2020_hiddenness}
Z.~Zhang, R.~Deng, D.~K.~Y. Yau, P.~Cheng, and J.~Chen, ``On hiddenness of
  moving target defense against false data injection attacks on power grid,''
  \emph{ACM Trans. Cyber-Physical Systems}, vol.~4, no.~3, pp. 1--29, Mar.
  2020.

\bibitem{xu_2022_robust}
W.~Xu, I.~M. Jaimoukha, and F.~Teng, ``Robust moving target defence against
  false data injection attacks in power grids,'' \emph{IEEE Trans. Inf.
  Forensics Secur.}, vol.~18, pp. 29--40, Sep. 2022.

\bibitem{hingorani_2007_facts}
N.~G. Hingorani, ``{FACTS} technology-state of the art, current challenges and
  the future prospects,'' in \emph{Proc. IEEE Power Engineering Society General
  Meeting}, Tampa, FL, USA, Jun. 2007.

\bibitem{zhang_2012_flexible}
X.~Zhang, C.~Rehtanz, and B.~Pal, \emph{Flexible AC Transmission Systems:
  Modelling and Control}.\hskip 1em plus 0.5em minus 0.4em\relax Springer
  Science \& Business Media, 2012.

\bibitem{morrow_2012_topology}
K.~L. Morrow, E.~Heine, K.~M. Rogers, and T.~J. Bobba, R. B .and~Overbye,
  ``Topology perturbation for detecting malicious data injection,'' in
  \emph{Proc. 45th Hawaii Int. Conf. on System Sciences}, Maui, HI, USA, Jan.
  2012, pp. 2104--2113.

\bibitem{wang_2015_effects}
S.~Wang, W.~Ren, and U.~M. Al-Saggaf, ``Effects of switching network topologies
  on stealthy false data injection attacks against state estimation in power
  networks,'' \emph{IEEE Sensors J.}, vol.~11, no.~4, pp. 2640--2651, Nov.
  2015.

\bibitem{brown_2020_transmission}
W.~E. Brown and E.~Moreno-Centeno, ``Transmission-line switching for load shed
  prevention via an accelerated linear programming approximation of ac power
  flows,'' \emph{IEEE Trans. Power Systems}, vol.~35, no.~4, pp. 2575--2585,
  Jul. 2020.

\bibitem{miao_2024_robust}
Z.~Miao and J.~Yu, ``A robust learning framework for smart grids in defense
  against false-data injection attacks,'' \emph{ACM Trans. Sensor Networks},
  vol.~20, no.~2, pp. 1--12, Jan. 2024.

\bibitem{heydari_2018_moving}
V.~Heydari, ``Moving target defense for securing scada communications,''
  \emph{IEEE Access}, vol.~6, pp. 33\,329--33\,343, Jun. 2018.

\bibitem{hu_2021_network}
Y.~Hu, P.~Xun, P.~Zhu, Y.~Xiong, Y.~Zhu, W.~Shi, and C.~Hu, ``Network-based
  multidimensional moving target defense against false data injection attack in
  power system,'' \emph{Computers \& Security}, vol. 107, p. 102283, Aug. 2021.

\bibitem{lakshminarayana_2024_survey}
\BIBentryALTinterwordspacing
S.~Lakshminarayana, Y.~Chen, C.~Konstantinou, D.~Mashima, and A.~K. Srivastava,
  ``Survey of moving target defense in power grids: Design principles,
  tradeoffs, and future directions,'' 2024. [Online]. Available:
  \url{https://arxiv.org/abs/2409.18317}
\BIBentrySTDinterwordspacing

\bibitem{rahman_false_2012}
M.~A. Rahman and H.~Mohsenian-Rad, ``False data injection attacks with
  incomplete information against smart power grids,'' in \emph{Proc. {IEEE}
  {Global} {Commun.} {Conf.}}, Anaheim, CA, USA, Dec. 2012, pp. 3153--3158.

\bibitem{li_blind_2013}
X.~Li, H.~V. Poor, and A.~Scaglione, ``Blind topology identification for power
  systems,'' in \emph{Proc. {IEEE} {Int.} {Conf.} on {Smart} {Grid} {Commun.}},
  Vancouver, Canada, Oct. 2013, pp. 91--96.

\bibitem{yu_blind_2015}
Z.~H. Yu and W.~L. Chin, ``Blind {false} {data} {injection} {attack} {using}
  {PCA} {approximation} {method} in {smart} {grid},'' \emph{IEEE Trans. Smart
  Grid}, vol.~6, no.~3, pp. 1219--1226, May 2015.

\bibitem{sun_2019_learning}
K.~Sun, I.~Esnaola, A.~M. Tulino, and H.~V. Poor, ``Learning requirements for
  stealth attacks,'' in \emph{Proc. IEEE Int. Conf. on Acoust., Speech and
  Signal Process.}, Brighton, United Kingdom, 2019, pp. 8102--8106.

\bibitem{sun_2023_asymptotic}
------, ``Asymptotic learning requirements for stealth attacks on linearized
  state estimation,'' \emph{IEEE Trans. Smart Grid}, vol.~14, no.~4, pp. 3189
  -- 3200, Jul. 2023.

\bibitem{grainger_power_1994}
J.~J. Grainger and W.~D. Stevenson, \emph{Power {S}ystem {A}nalysis}.\hskip 1em
  plus 0.5em minus 0.4em\relax McGraw-Hill, 1994.

\bibitem{abur_power_2004}
A.~Abur and A.~G. Exp{\'o}sito, \emph{Power {System} {State} {Estimation}:
  {Theory} and {Implementation}}.\hskip 1em plus 0.5em minus 0.4em\relax CRC
  Press, Mar. 2004.

\bibitem{kay_1993_fundamentals}
S.~M. Kay, \emph{Fundamentals of Statistical Signal Processing: Estimation
  Theory}.\hskip 1em plus 0.5em minus 0.4em\relax Prentice-Hall, Inc., 1993.

\bibitem{kettner_2017_properties}
A.~M. Kettner and M.~Paolone, ``On the properties of the power systems nodal
  admittance matrix,'' \emph{IEEE Trans. Power Systems}, vol.~33, no.~1, pp.
  1130--1131, Jan. 2018.

\bibitem{zhuang_2014_towards}
R.~Zhuang, S.~A. DeLoach, and X.~Ou, ``Towards a theory of moving target
  defense,'' in \emph{Proc. of the first ACM Workshop on Moving Target
  Defense}, Scottsdale, Arizona, USA, Nov. 2014, pp. 31--40.

\bibitem{rogers_2008_some}
K.~M. Rogers and T.~J. Overbye, ``Some applications of distributed flexible ac
  transmission system {(D-FACTS)} devices in power systems,'' in \emph{Proc.
  40th North American Power Symp.}, Calgary, Canada, Sep. 2008, pp. 1--8.

\bibitem{west_2001_introduction}
D.~B. West, \emph{Introduction to Graph Theory}.\hskip 1em plus 0.5em minus
  0.4em\relax Pearson, 2017.

\bibitem{graham_1985_history}
R.~L. Graham and P.~Hell, ``On the history of the minimum spanning tree
  problem,'' \emph{Annals of the History of Computing}, vol.~7, no.~1, pp.
  43--57, 1985.

\bibitem{zimmerman_matpower:_2011}
R.~D. Zimmerman, C.~E. Murillo-S\'{a}nchez, and R.~J. Thomas, ``{MATPOWER}:
  {Steady}-{state} {operations}, {planning}, and {analysis} {tools} for {power}
  {systems} {research} and {education},'' \emph{IEEE Trans. Power Syst.},
  vol.~26, no.~1, pp. 12--19, Feb. 2011.

\bibitem{leon_2020_quadratically}
L.~M. Leon, A.~S. Bretas, and S.~Rivera, ``Quadratically constrained quadratic
  programming formulation of contingency constrained optimal power flow with
  photovoltaic generation,'' \emph{Energies}, vol.~13, no.~13, p. 3310, Jun.
  2020.

\bibitem{li_2011_probabilistic}
W.~Li, \emph{Probabilistic Transmission System Planning}.\hskip 1em plus 0.5em
  minus 0.4em\relax John Wiley \& Sons, 2011.

\end{thebibliography}

\appendices

\section{Proof of Lemma \ref{lemma:rank}}\label{Sec:Proof0}
We begin by noting that 
\begin{IEEEeqnarray}{c}
\begin{bmatrix}
\Am^{\sf T} \Dm \Am \\
 \Dm \Am \\
 - \Dm \Am 
\end{bmatrix} =
 \begin{bmatrix}
\Am^{\sf T}  \\
 \Id \\
 - \Id  
\end{bmatrix}
 \Dm \Am ,
\end{IEEEeqnarray}
where $\Id$ is the identity matrix of dimension $l$. 
Furthermore, it holds that any row of $\Am^{\sf T}$ can be expressed as a linear combination of the rows of $\Id$. 
Given the fact that the rank of an identity matrix is  given by its dimension, it holds that 
\begin{IEEEeqnarray}{c}
 \textnormal{rank} 
\left( \begin{bmatrix}
\Am^{\sf T}  \\
 \Id \\
 - \Id  
\end{bmatrix} \right) =l.
\end{IEEEeqnarray}

As a result, it holds that 
\begin{IEEEeqnarray}{c}
\textnormal{rank} 
\left( 
\begin{bmatrix}
\Am^{\sf T} \Dm \Am \\
 \Dm \Am \\
 - \Dm \Am 
\end{bmatrix}
\right) = 
\textnormal{rank} 
\left( 
 \begin{bmatrix}
\Am^{\sf T}  \\
 \Id \\
 - \Id  
\end{bmatrix}
 \Dm \Am 
\right)
= \textnormal{rank} \left(  \Dm \Am \right). \IEEEeqnarraynumspace \squeezeequ
\end{IEEEeqnarray}
The proof is completed by noticing that $\Dm$ is a full rank matrix and the rank of $ \Am$ is $n$ \cite[Lemma 1]{kettner_2017_properties}, which implies that 
\begin{IEEEeqnarray}{c}
\textnormal{rank} \left(  \Dm \Am \right) = \textnormal{rank} \left(  \Am \right) = n.
\end{IEEEeqnarray}

\section{Proof of Lemma \ref{lemma:N1}}\label{Sec:Proof1}
We describe the diagonal branch admittance matrix as
\begin{IEEEeqnarray}{c}
\Dm =  \sum_{i=1}^l b_i \Em_{ii},
\end{IEEEeqnarray}
where $\Em_{ii} \in \Rc^{l \times l}$ is a diagonal matrix with zeros in all entries with the exception of the $i$-th diagonal entry, which is one.
It holds that 
\begin{IEEEeqnarray}{rl}
\Dm \Am &= \sum_{i=1}^l b_i \Em_{ii} \Am, \label{Equ:temp_1}
\end{IEEEeqnarray}
and therefore 
\begin{IEEEeqnarray}{rl}
\Am^{\sf T}\Dm \Am &= \Am^{\sf T} \left( \sum_{i=1}^l b_i \Em_{ii} \Am\right) = \sum_{i=1}^l b_i  \Am^{\sf T} \Em_{ii} \Am, \label{Equ:temp_2}
\end{IEEEeqnarray}
where $\Em_{ii}\Am \in \Rc^{l \times (n+1)}$ is matrix of zeros with the exception of the $i$-th row, which is the row vector $\Am_{i, \cdot}$; 
and $\Am^{\sf T} \Em_{ii} \Am \in \Rc^{(n+1) \times (n+1)}$ is given by
\begin{IEEEeqnarray}{c}
\Am^{\sf T} \Em_{ii} \Am =
\begin{bmatrix}
a_{i,1}  \Am_{i, \cdot}\\
\vdots \\
a_{i,n} \Am_{i, \cdot}
\end{bmatrix}.
\end{IEEEeqnarray}
Note that the $i$-th row of $\Am$, i.e.  $\Am_{i, \cdot}$, only has two nonzero entries (either $1$ or $-1$), which are given by 
\begin{IEEEeqnarray}{c}
\Am_{i, \cdot} = 
 \begin{blockarray}{ccccc}
   \dots & f(\mathsf{e}_i) & \dots & t(\mathsf{e}_i)  \\
 \begin{block}{[ccccc]}
  0 &1 &\cdots &-1 & 0  \\
 \end{block}
 \end{blockarray}.
 \end{IEEEeqnarray}
 Hence, from \eqref{Equ:temp_1} and \eqref{Equ:temp_2}, it holds that 
 \begin{IEEEeqnarray}{c}
\begin{bmatrix}
\Am^{\sf T} \Dm \Am \\
 \Dm \Am \\
 - \Dm \Am 
\end{bmatrix}
= \sum_{i=1}^l b_i
\begin{bmatrix}
\Am^{\sf T} \Em_{ii} \Am \\
\Em_{ii} \Am\\
 - \Em_{ii} \Am
\end{bmatrix}
=\sum_{i=1}^l b_i \Hm_i. \label{Equ:temp3}
 \end{IEEEeqnarray}

The proof is completed by combining \eqref{Equ:JacoPhyFR} with \eqref{Equ:temp3}.

%
%
%

\section{Proof of Theorem \ref{Theorem:2}}\label{Sec:Proof3}

If it holds that ${\cal E_{MTD}} = \{ \mathsf{e}_k\}$, from Lemma \ref{lemma:N1}, the Jacobian matrix in \eqref{Equ:Jac_MTD} satisfies 
\begin{IEEEeqnarray}{c}
\Hm' = \Hm + \Delta b_k \left[\Hm_k\right]_{\cdot, \{2, \ldots, n+1 \}}, \label{Equ:SBHchange}
\end{IEEEeqnarray}
where $\Hm_k$ is defined in \eqref{Equ:HmK}.

Consider the case for which $f(\mathsf{e}_k) \neq 1$ and $t(\mathsf{e}_k) \neq 1$ first. 
For all $\cv \in  \{ \Rc^{n}{\setminus}\zerov\}$ that satisfy the condition in \eqref{Equ:SBMTDCon1}, 
it follows that \eqref{Equ:SBHchange} implies
\begin{IEEEeqnarray}{rl}
 \Hm \cv & = \left(\Hm' - \Delta b_k \left[\Hm_k\right]_{\cdot, \{2, \ldots, n+1 \}}\right) \cv \\
 &= \Hm' \cv  - \Delta b_k \left[\Hm_k\right]_{\cdot, \{2, \ldots, n+1 \}}\cv \\
& = \Hm' \cv, \label{Equ:Key}
\end{IEEEeqnarray}
where the equality in \eqref{Equ:Key} follows from the fact that the matrix $\Hm_k$ in \eqref{Equ:HmK} satisfies 
\begin{IEEEeqnarray}{c}
\left[\Hm_k\right]_{\cdot, \{2, \ldots, n+1 \}} \cv = \zerov
\end{IEEEeqnarray}
for $f(\mathsf{e}_k) \neq 1$ and $t(\mathsf{e}_k) \neq 1$. 

Hence, for the attack constructed as in \eqref{Equ:SBMTD}, it holds that
\begin{IEEEeqnarray}{rl}\label{Equ:Key1}
r'(\yv_a) = r'(\yv + \av) = r'(\yv + \Hm\cv) &= r'(\yv + \Hm'\cv) \label{Equ:Key2}\\
&= r'(\yv),\label{Equ:Key3}
\end{IEEEeqnarray}
where $r'$ is defined in \eqref{Equ:ResN}; the equality in \eqref{Equ:Key2} follows from \eqref{Equ:Key}; 
and the equality in \eqref{Equ:Key3}  follows from Theorem \ref{Theorem:Liu}.
Furthermore, it holds that 
\begin{IEEEeqnarray}{c}
\Km' \yv_a = \Km' \left( \yv + \Hm\cv\right) =  \Km' \left( \yv + \Hm'\cv\right) = \Km'  \yv +\cv,\IEEEeqnarraynumspace
\end{IEEEeqnarray}
where $\Km'$ is defined in \eqref{Equ:WLSN}.
Hence, for all $\cv \in \{ \Rc^{n}{\setminus}\zerov\}$ such that the condition in \eqref{Equ:SBMTDCon1} holds, we have that 
\begin{IEEEeqnarray}{c}
\Km' \yv_a \neq \Km'  \yv.\label{Equ:Key4}
\end{IEEEeqnarray}

Note that \eqref{Equ:Key3} and \eqref{Equ:Key4} comply with Definition \ref{Def:SMTD}.
As a result, for all $\cv \in  \{ \Rc^{n}{\setminus}\zerov\}$ such that the condition in \eqref{Equ:SBMTDCon1} holds, the attacks constructed as \eqref{Equ:SBMTD} are MTD-resilient stealth attacks under the given Single Branch MTD strategy when $f(\mathsf{e}_k) \neq 1$ and $t(\mathsf{e}_k) \neq 1$.

Then, consider the case in which $f(\mathsf{e}_k) = 1$ or $t(\mathsf{e}_k) = 1$. 
For all $\cv \in  \{ \Rc^{n}{\setminus}\zerov\}$ such that the condition in \eqref{Equ:SBMTDCon2} holds, 
from \eqref{Equ:SBHchange}, it holds that 
\begin{IEEEeqnarray}{rl}
 \Hm \cv & = \left(\Hm' - \Delta b_k \left[\Hm_k\right]_{\cdot, \{2, \ldots, n+1 \}}\right) \cv \\
 &= \Hm' \cv  - \Delta b_k \left[\Hm_k\right]_{\cdot, \{2, \ldots, n+1 \}}\cv \label{Equ:OY}\\
& = \Hm' \cv, \label{Equ:OY1}
\end{IEEEeqnarray}
where the matrix $\left[\Hm_k\right]_{\cdot, \{2, \ldots, n+1 \}}$ in \eqref{Equ:OY} is a matrix that either the $(f(\mathsf{e}_k)-1)$-th or the $(t(\mathsf{e}_k)-1)$-th column is the only nonzero column; 
the equality in \eqref{Equ:OY1} follows from the fact that
\begin{IEEEeqnarray}{c}
\left[\Hm_k\right]_{\cdot, \{2, \ldots, n+1 \}} \cv = \zerov. \label{Equ:OY4}
\end{IEEEeqnarray}
Hence, for the attack constructed as \eqref{Equ:SBMTD}, it holds that
\begin{IEEEeqnarray}{rl}\label{Equ:Key1}
r'(\yv_a) = r'(\yv + \av) = r'(\yv + \Hm\cv) &= r'(\yv + \Hm'\cv) \label{Equ:OY2}\\
&= r'(\yv),\label{Equ:OY3}
\end{IEEEeqnarray}
where $r'$ is defined in \eqref{Equ:ResN}, the equality in \eqref{Equ:OY2} follows from \eqref{Equ:OY1}; 
and the equality in \eqref{Equ:OY3}  follows from Theorem \ref{Theorem:Liu}.

Furthermore,  for all $\cv \in \{ \Rc^{n}{\setminus}\zerov\}$ such that the condition in \eqref{Equ:SBMTDCon2} holds, we have that 
\begin{IEEEeqnarray}{c}
\Km' \yv_a = \Km' \left( \yv + \Hm\cv\right) =  \Km' \left( \yv + \Hm'\cv\right) = \Km'  \yv +\cv \neq  \Km'  \yv ,\IEEEeqnarraynumspace \label{Equ:OY6} \squeezeequ
\end{IEEEeqnarray}
where $\Km'$ is defined in \eqref{Equ:WLSN}.

Note that \eqref{Equ:OY3} and \eqref{Equ:OY6} comply with Definition \ref{Def:SMTD}.
As a result, for all $\cv \in  \{ \Rc^{n}{\setminus}\zerov\}$ such that the condition in \eqref{Equ:SBMTDCon2} holds, the attacks constructed as \eqref{Equ:SBMTD} are MTD-resilient stealth attacks under the Jacobian matrix $\Hm'$ when $f(\mathsf{e}_k) = 1$ or $t(\mathsf{e}_k) = 1$.

This completes the proof. 
\section{Proof of Lemma \ref{lemma2}}\label{Sec:Proof4}

Consider the attack $\av$ constructed via Theorem \ref{Theorem:2}, from \eqref{Equ:Key} and \eqref{Equ:OY1}, it holds that 
\begin{IEEEeqnarray}{c}
\Hm\cv = \Hm'\cv. 
\end{IEEEeqnarray}
Under the Single Branch MTD and the attack in \eqref{Equ:SBMTD}, the measurement model is given by 
\begin{IEEEeqnarray}{c}
\yv_a = \Hm' \xv + \zv + \Hm'\cv.
\end{IEEEeqnarray}

Consider the case in which $f(\mathsf{e}_k) \neq 1$ and $t(\mathsf{e}_k) \neq 1$ first. 
From \eqref{Equ:Jac_MTD} and Lemma \ref{lemma:N1}, the $(n+1+k)$-th row of matrix $\Hm'$ is given by 
\begin{IEEEeqnarray}{c}
\Hm_{n+1+k, \cdot} = \left[ 0, \ldots, 1 \ldots, -1, \ldots \right], 
\end{IEEEeqnarray}
in which the nonzero entries are the $(f(\mathsf{e}_k)-1)$-th entry and  the $(t(\mathsf{e}_k)-1)$-th entry. 
Hence,  for any $\cv \in \Rc^{n}$ satisfying \eqref{Equ:SBMTDCon1}, 
it holds that 
\begin{IEEEeqnarray}{c}
a_{n+1+k} = \Hm_{n+1+k, \cdot} \cv =0. 
\end{IEEEeqnarray}

Then consider the case in which $f(\mathsf{e}_k) = 1$ or $t(\mathsf{e}_k) = 1$.
From \eqref{Equ:Jac_MTD} and Lemma \ref{lemma:N1}, the $(n+1+k)$-th row of matrix $\Hm'$ is given by 
\begin{IEEEeqnarray}{rl}
\Hm_{n+1+k, \cdot} = \left[ 0, \ldots, -1 \ldots, 0 \right] & \ \textnormal{ if } \ f(\mathsf{e}_k) = 1, \\
\Hm_{n+1+k, \cdot} = \left[ 0, \ldots, 1 \ldots, 0\right] & \ \textnormal{ if } \ t(\mathsf{e}_k) = 1, 
\end{IEEEeqnarray}
in which the nonzero entry is the $(t(\mathsf{e}_k) -1)$-th entry if $f(\mathsf{e}_k) = 1$, and is the $(f(\mathsf{e}_k) -1)$-th entry if $t(\mathsf{e}_k) = 1$.
Hence,  for any $\cv \in \Rc^{n}$ satisfying \eqref{Equ:SBMTDCon2}, 
it holds that 
\begin{IEEEeqnarray}{c}
a_{n+1+k} = \Hm_{n+1+k, \cdot} \cv =0. 
\end{IEEEeqnarray}

The proof is completed by noticing that the $(n+1+l+k)$-th element of $\av$ is just the negative of $(n+1+k)$-th element.

\section{Proof of Theorem \ref{Theorem:3}}\label{Sec:Proof5}

Without loss of generality, we assume that the first $k$ ($k \leq l$) branches are protected by the MTD strategy, i.e. 
\begin{IEEEeqnarray}{c}
{\cal E_{MTD}} = \{ \mathsf{e}_1, \ldots, \mathsf{e}_k\}.
\end{IEEEeqnarray}
Similarly as \eqref{Equ:SBHchange}, it holds that
\begin{IEEEeqnarray}{c}
\Hm' = \Hm + \sum_{i=1}^k \Delta b_i \left[\Hm_k\right]_{\cdot, \{2, \ldots, n+1 \}}, 
\end{IEEEeqnarray}
where $\Hm_{i}$ is given in \eqref{Equ:HmK}.

Consider the case in which the condition in \eqref{Equ:MBMTDCon0} holds. 
For all $\cv \in \{ \Rc^{n}{\setminus}\zerov\}$ such that the condition in \eqref{Equ:MBMTDCon1} holds, it holds that 
\begin{IEEEeqnarray}{c}
\Hm \cv = \Hm' \cv  - \sum_{i=1}^k \Delta b_i \left[\Hm_i \right]_{\cdot, \{2, \ldots, n+1 \}}\cv = \Hm' \cv, \label{Equ:MBKey}\IEEEeqnarraynumspace
\end{IEEEeqnarray}
where the last equality follows from the fact that for all $i \in \{1, \ldots,k \}$,
\begin{IEEEeqnarray}{c}
\left[\Hm_i\right]_{\cdot, \{2, \ldots, n+1 \}} \cv = \zerov .
\end{IEEEeqnarray}
Hence, for the attack constructed as \eqref{Equ:MBMTD}, it holds that
\begin{IEEEeqnarray}{c}\label{Equ:MBKey1}
r'(\yv_a) = r'(\yv + \av) = r'(\yv + \Hm'\cv) = r'(\yv) ,
\end{IEEEeqnarray}
where $r'$ is defined in \eqref{Equ:ResN}, and the equality in \eqref{Equ:MBKey1} follows from \eqref{Equ:MBKey} and Theorem \ref{Theorem:Liu}.

From \eqref{Equ:WLSN}, it holds that 
\begin{IEEEeqnarray}{c}
\Km' \yv_a = \Km' \left( \yv + \Hm\cv\right) =  \Km' \left( \yv + \Hm'\cv\right) = \Km'  \yv +\cv. \IEEEeqnarraynumspace
\end{IEEEeqnarray}
Hence, for all $\cv \in \{ \Rc^{n}{\setminus}\zerov\}$ such that the condition in \eqref{Equ:MBMTDCon1} holds, it holds that 
\begin{IEEEeqnarray}{c}
\Km' \yv_a \neq \Km'  \yv.\label{Equ:MBKey4}
\end{IEEEeqnarray}

As a result, for all $\cv \in \{ \Rc^{n}{\setminus}\zerov\}$ such that the condition in \eqref{Equ:MBMTDCon1} holds, the attacks constructed as \eqref{Equ:MBMTD} are MTD-resilient stealth attacks under the given Multiple Branch MTD strategy.

Similarly, for the case in which the condition in \eqref{Equ:MBMTDCon3N} holds, if vector $\cv \in \{ \Rc^{n}{\setminus}\zerov\}$ satisfies the condition in \eqref{Equ:MBMTDCon2}, it holds that 
\begin{IEEEeqnarray}{c}
\Hm \cv = \Hm' \cv, \quad r'(\yv_a) = r'(\yv), \quad \textnormal{and} \quad \Km' \yv_a \neq \Km'  \yv, \IEEEeqnarraynumspace
\end{IEEEeqnarray}
which implies that the attacks constructed as \eqref{Equ:MBMTD} are MTD-resilient stealth attacks under the given Multiple Branch MTD strategy.

This completes the proof.

\section{Proof of Lemma \ref{lemma3}}\label{Sec:Proof6}
Consider the case in which the condition in \eqref{Equ:MBMTDCon0} holds. 
Note that if the vector $\cv \in \{ \Rc^{n}{\setminus}\zerov\}$ satisfies \eqref{Equ:MBMTDCon1}, then $\cv$ also satisfies \eqref{Equ:SBMTDCon1}. 
Then from Lemma \ref{lemma2}, for all $ \mathsf{e} \in {\cal E_{MTD}}$, it holds that 
\begin{IEEEeqnarray}{c}
a_{n+1+b(\mathsf{e})} = 0 \quad \textnormal{and} \quad a_{n+1+l+b(\mathsf{e})} = 0.
\end{IEEEeqnarray}

Then consider the case in which the condition in \eqref{Equ:MBMTDCon3N} holds.
For all $ \mathsf{e} \in {\cal E_{MTD}}$ such that $ f(\mathsf{e}_k) = 1$ or $t(\mathsf{e}_k) = 1$, if the vector $\cv \in \{ \Rc^{n}{\setminus}\zerov\}$ satisfies \eqref{Equ:MBMTDCon2}, then $\cv$ also satisfies \eqref{Equ:SBMTDCon2}. 
Then from Lemma \ref{lemma2}, it holds that 
\begin{IEEEeqnarray}{c}
a_{n+1+b(\mathsf{e})} = 0 \quad \textnormal{and} \quad a_{n+1+l+b(\mathsf{e})} = 0.
\end{IEEEeqnarray}
For the rest branches, i.e. $ \mathsf{e} \in {\cal E_{MTD}}$ such that $ f(\mathsf{e}_k) \neq 1$ and $t(\mathsf{e}_k) \neq 1$, if the vector $ \cv \in \{ \Rc^{n}{\setminus}\zerov\}$ satisfies \eqref{Equ:MBMTDCon2}, then $\cv$ also satisfies \eqref{Equ:SBMTDCon1}, which is already considered the first case. 

This completes the proof. 

\section{Proof of Lemma \ref{lemma4}}\label{Sec:Proof7}
Note that a spanning tree of an undirected graph is a tree which includes all of the vertices $\Vc$.
As a result, if $\Gc'$ forms a spanning tree for $\Gc$, the reference bus is definitely connected with some bus in $\Gc'$. 
To that end, the condition in \eqref{Equ:MBMTDCon0} does not holds. 
Hence, only the case in which the condition in \eqref{Equ:MBMTDCon3N} holds needs to be considered. 

When $\Gc'$ is a spanning tree of $\Gc$, it holds that any two vertices (or buses) are connected by a path.
Hence, the reference bus can be connected with the rest of the buses via a path, in which it holds for any branch $\mathsf{e}$ in this path that 
\begin{subequations}
\begin{IEEEeqnarray}{rl}
c_{f(\mathsf{e})-1} = 0, \quad &\textnormal{ if} \ t(\mathsf{e}) \ = 1,  \\
c_{t(\mathsf{e})-1} \ = 0, \quad &\textnormal{ if} \ f(\mathsf{e}) = 1,\\
c_{f(\mathsf{e})-1} = c_{t(\mathsf{e})-1} = 0,   \quad &\textnormal{ if} \ f(\mathsf{e}) \ \neq 1 \textnormal{ and }\ t(\mathsf{e}) \ \neq 1. \IEEEeqnarraynumspace 
\end{IEEEeqnarray}
\end{subequations}
As a result, the conditions in \eqref{Equ:MBMTDCon21} - \eqref{Equ:MBMTDCon23} hold only for $\cv = \zerov$, which is against the condition $\cv\neq\zerov$ in Theorem \ref{Theorem:3}. 
Hence, if the subgraph $\Gc'(\Vc,{\cal E_{MTD}})$ is a spanning tree of $\Gc(\Vc,\Ec)$, there does not exist a vector $\cv$ satisfying the condition in \eqref{Equ:MBMTDCon2}.

This completes the proof. 

\end{document}